\documentclass[conference]{IEEEtran}
\IEEEoverridecommandlockouts

\usepackage{cite}
\usepackage{amsmath,amssymb,amsfonts}
\usepackage{multirow}

\usepackage{algorithmic}
\usepackage{graphicx}
\usepackage{textcomp}
\usepackage{xcolor}

\usepackage{cancel}
\newtheorem{lemma}{Lemma}
\newtheorem{Corollary}{Corollary}
\newtheorem{proof}{Proof}
\usepackage{setspace}

\usepackage{natbib}
\setcitestyle{numbers,square}

\def\BibTeX{{\rm B\kern-.05em{\sc i\kern-.025em b}\kern-.08em
    T\kern-.1667em\lower.7ex\hbox{E}\kern-.125emX}}
\begin{document}


\title{A Novel Channel-Constrained Model for 6G Vehicular Networks with Traffic Spikes}

\author{

\IEEEauthorblockN{1\textsuperscript{st} Ke Deng, 2\textsuperscript{nd} Zhiyuan He, 3\textsuperscript{rd} Haohan Lin, 4\textsuperscript{th} Hao Zhang, 5\textsuperscript{th} Desheng Wang}
\IEEEauthorblockA{\textit{Wuhan National Laboratory for Optoelectronics} \\
\textit{Huazhong University of Science and Technology}\\
Wuhan, China \\
sheepsui\@hust.edu.cn, zedyuanhe\@hust.edu.cn, M202272477\@hust.edu.cn, haoz52\@hust.edu.cn, dswang\@hust.edu.cn}





}

\maketitle

\begin{abstract}
Mobile Edge Computing (MEC) holds excellent potential in Congestion Management (CM) of 6G vehicular networks. A reasonable schedule of MEC ensures a more reliable and efficient CM system. Unfortunately, existing parallel and sequential models cannot cope with scarce computing resources and constrained channels, especially during traffic rush hour. In this paper, we propose a channel-constrained multi-core sequential model (CCMSM) for task offloading and resource allocation. The CCMSM incorporates a utility index that couples system energy consumption and delay, applying Genetic Algorithm combining Sparrow Search Algorithm (GA-SSA) in the branching optimization. Furthermore, we prove that the system delay is the shortest with the FCFS computing strategy in the MEC server. Simulation demonstrates that the proposed CCMSM achieves a higher optimization level and exhibits better robustness and resilient scalability for traffic spikes. \end{abstract}


\section{Introduction}
The emergence of 6G communication technology is expected to further enhance the capabilities of vehicular networks by offering ultra-high data rates, ultra-low latency, and massive device connectivity which enable a plethora of novel applications and services, such as autonomous driving, real-time traffic management, and advanced vehicle-to-everything (V2X) communication\cite{b1}. 
However, the increasing demand for high-bandwidth applications especially CM poses the challenge of channel constraints, particularly during traffic rush hour\cite{b1, b2}. Given the high granularity of CM tasks, MEC is required to divide the service of data collection including analysis into sub-tasks and distribute them across computing nodes, ensuring efficient processing and responding to the dynamic of traffic conditions\cite{b35}. Despite the benefits of MEC in reducing computing and transmission costs and promoting the development of environment-friendly vehicular networks\cite{b36, b37}, task scheduling in MEC presents challenges due to various factors, such as disparities in base station access and computing capabilities, spatial heterogeneity of terminals, and diverse task demands\cite{b4}. 

Therefore, an urgent need exists for practical scheduling models that are tailored to real-world vehicular network scenarios. Many researchers currently focus on task offloading and resource allocation in MEC. In previous work\cite{b14}, the parallel scheduling model is adopted to simplify the problem. However, the parallel model assumes an unconstrained server cluster where simultaneous workloads are not limited, which contradicts the finite number of MEC cores and indivisible computing frequencies in reality. The latest research \cite{b20} introduces a multi-core sequential scheduling model, which is more suitable for computing hardware. However, the traffic spike caused by transportation gridlock challenges the reality of limited channel resources in CM\cite{b6}. This constraint leads to a necessary queue for task offloading, introducing a complex coupling relationship between it and the computing queue in MEC. Overall, the existing model falls short of accurately modeling the intricacies of MEC in 6G vehicular networks, rendering it far from deriving relevant insights for deployment guidance. 

Furthermore, a comprehensive task scheduling model in CM necessitates an effective algorithm that is tailored to the heterogeneous optimization variable. 
Reference \cite{b22} demonstrates the great potential of the GA in rationally determining offloading strategies in MEC, proving the GA's feasibility in task scheduling. However, GA is not suitable for continuous variables \cite{b22}. Regardless of that Mohamed \cite{b24} et al. applied the SSA to the distribution network's performance optimization, proving its feasibility and accuracy in optimizing continuous power variables, the queue of task offloading and computing brings the optimization coupled constraints.

Generally, to advance the deployment of 6G vehicular networks through optimal task scheduling in MEC, we make the following contributions:

\begin{itemize}
\item We propose the CCMSM by introducing an offloading queue of terminal tasks, which is essential for MEC deployment with a single base station and multiple terminals in CM.
\item A novel utility index that balances the energy consumption and delay in optimization was proposed by introducing a preference coefficient. 
\item A series of experiments were conducted to demonstrate that Genetic Algorithms - Sparrow Search Algorithms (GA-SSA) can serve as a competent algorithm for CCMSM.
\item We proved that the system delay is the shortest when the computing queue of the MEC core is given by the FCFS criterion to simplify and accelerate the algorithm.
\end{itemize}

The rest of this paper is as follows. Section II presents the system model, followed by the formulation of a combinatorial optimization problem. The algorithmic implementation is discussed in Section III. Section IV validates the proposed CCMSM through a series of simulation experiments, demonstrating its feasibility and robustness. Finally, conclusions are drawn in Section V.

\section{System Model and Problem Formulation}\label{System Model and Problem Formulation}

\subsection{System Modelling}\label{Systems Modeling}
The architecture of the CCMSM is shown in Fig.~\ref{fig1}.

In CCMSM, there are $N$ terminals in the system, of which the $i^{th}$ terminal contains $M_i$ tasks, represented by task set $W_i=\{w_{i1},w_{i2},...,w_{ij},...,w_{iM_i}\}$. To facilitate mathematical modeling, we set the upper limit of task number for each terminal to $L$, that is, $M_i\le L$. 

Due to the received data after computation offloading being much less than the original data\cite{b25}, CCMSM only models the uplink. In the calculation offloading, the information of task $w_{ij}$ can be represented by the tuple $\{u_{ij}, c_{ij}\}$, where $u_{ij}$ is the data volume and $c_{ij}$ represents the required calculation cycle.

The MEC edge server comprises $K$ computing cores that handle the offloaded  task set $Q^k$ transferred by terminals. Conversely, tasks calculated in terminals can be regarded as settled on a non-existent core with a mark of $0$ in MEC, which can be expressed as the subtraction of the total task set and the offloaded task set, that is $Q^0=\bigcup_{i=1}^{N}W_i-\bigcup_{k=1}^{K}Q^k.$ Furthermore, it is noteworthy that the computation and transmission of tasks is sequential, which  is decided by the reality of channel resource and computing hardware.

\begin{figure}[htbp]
\centerline{\includegraphics[height=3.5cm,width=0.45\textwidth]{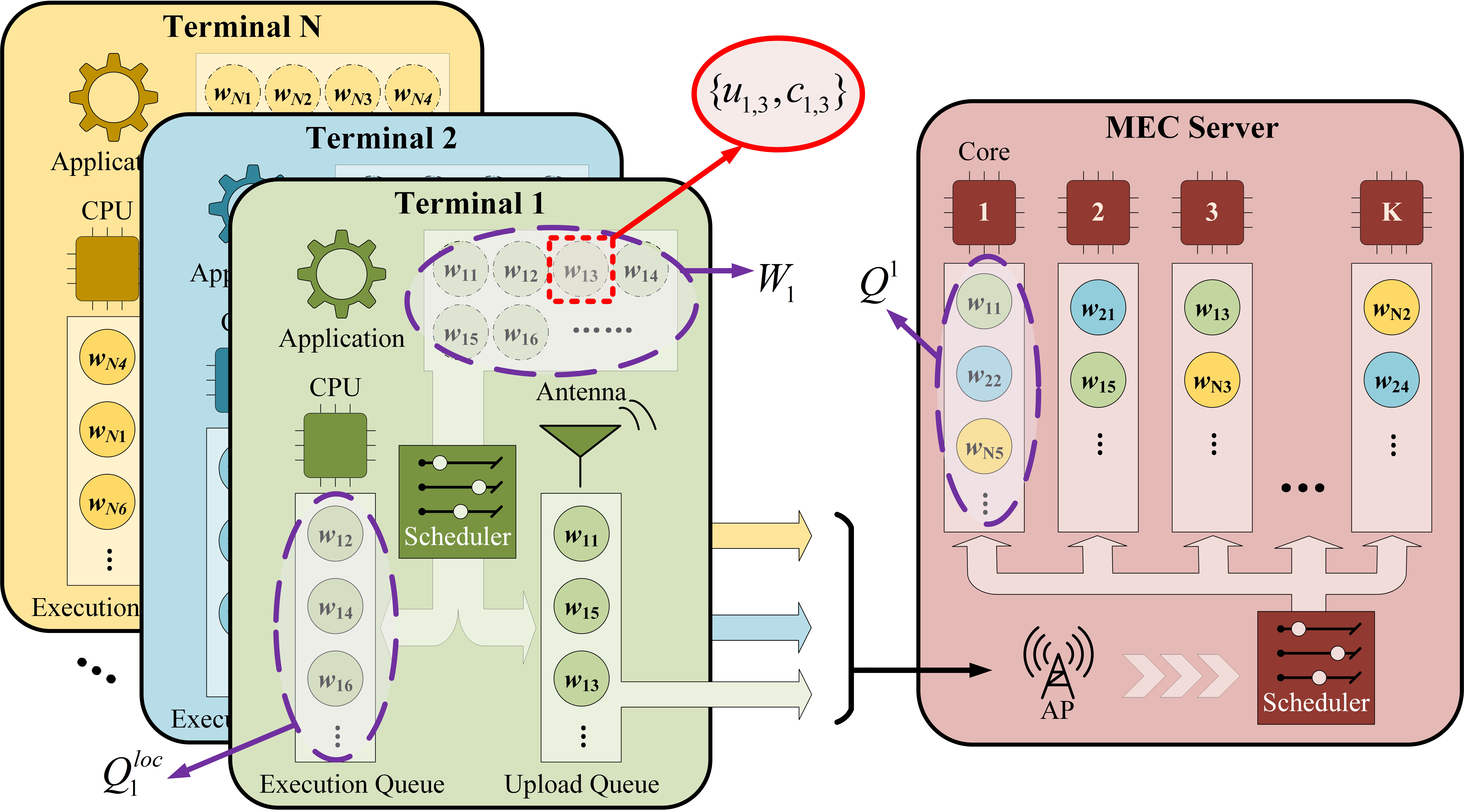}}
\caption{Channel-constrained multi-core sequential model.}
\label{fig1}
\end{figure}

\subsection{Optimization Goals}\label{Optimization Goals}

The essence of task offloading and resource allocation in MEC is seeking the optimal strategy to reduce the system delay and energy consumption. Therefore, we model these two parameters in section \ref{Delay} and \ref{Power Consumption}. In section \ref{Utility Index}, we design a utility index to couple them, transforming the task offloading and resource allocation in the CCMSM into a single-objective optimization. Due to the complexity of the optimization problem, we divide the overall optimization into two parts according to the greedy strategy and process them sequentially.

\subsubsection{Delay}\label{Delay}

In CCMSM, the delay is divided into the calculation and transmission part, in which computing delay of each task $t_{ij}^{exe}$ in terminal $i$ can be formulated as:

{\scriptsize
\begin{align}
  t_{ij}^{exe} =
  \begin{cases}
    c_{ij} / f_i^{loc}, & \forall w_{ij}\in Q^0 \\
    c_{ij} / f^{core}, & \forall w_{ij}\in\bigcup_{k=1}^{K}Q^k
  \end{cases}
  \label{eq1}
\end{align}
}

Where the $f_i^{loc}$ and $f^{core}$ are the CPU frequency of terminals and MEC cores. And the transfer delay of offloading task $\tau_{ij}^{ul}$ in terminal $i$ is formulated as:

{\scriptsize
\begin{align}
  \tau_{ij}^{ul} =
  \begin{cases}
    \frac{u_{ij}}{W^{ul}\log_2{(1+ \frac{p_i^{loc}h_i}{N_0W^{ul}})}} , & \forall w_{ij}\in\bigcup_{k=1}^{K}Q^k\\
    0, & \forall w_{ij}\in Q^0 
  \end{cases}
  \label{eq2}
\end{align}
}

Where $W^{ul}$ represents the bandwidth of the uplink, $p_i^{loc}$ is the transmission power of terminals, $h_i$ represents the channel gain of the uplink and $N_0$ sets the power spectral density of the channel noise. For systematic delay modelling, the transfer delay of tasks calculated in terminals is denoted as 0.

The offloaded tasks will be transferred to computing cores, whose arrival order will affect the total delay of cores, named offloading delay $T_k^{core}$. When modelling it, the arrival time $t_{ij}^{ul}$ is crucial which can be calculated through the transfer delay $\tau_{ij}^{ul}$ and the cumulative transfer delay of the pre-task determined by the offloading queue. We define the offloading queue of tasks in terminal $i$ is $\boldsymbol{q^{ul}_i}$ and record a vector $\boldsymbol{\tau_i}=(\tau_{i1}^{ul},\tau_{i2}^{ul},\ldots,\tau_{ij}^{ul},\ldots,\tau_{iL}^{ul})$ containing tasks’ transfer delay of terminal $i$. Besides, the queue function $F_{rank}$ is defined to sort the $\boldsymbol{\tau}i$ in the order of $\boldsymbol{q^{ul}_i}$. By cumulative summing up the rearranged vector, the arrival time $\boldsymbol{t_i}$ of tasks in terminal $i$ can be obtained as:

{\scriptsize
\begin{align}
\boldsymbol{t_i} = F_{rank}\left(\boldsymbol{\tau_i},\boldsymbol{q^{ul}_i} \right)^T \left(\begin{matrix}1&1&\ldots&1\\0&1&\ldots&1\\0&0&\ddots&\vdots\\0&0&0&1\\\end{matrix}\right)
  \label{eq3}
\end{align}
}
\begin{figure}[htbp]
\centerline{\includegraphics[height=2.4cm,width=0.49\textwidth]{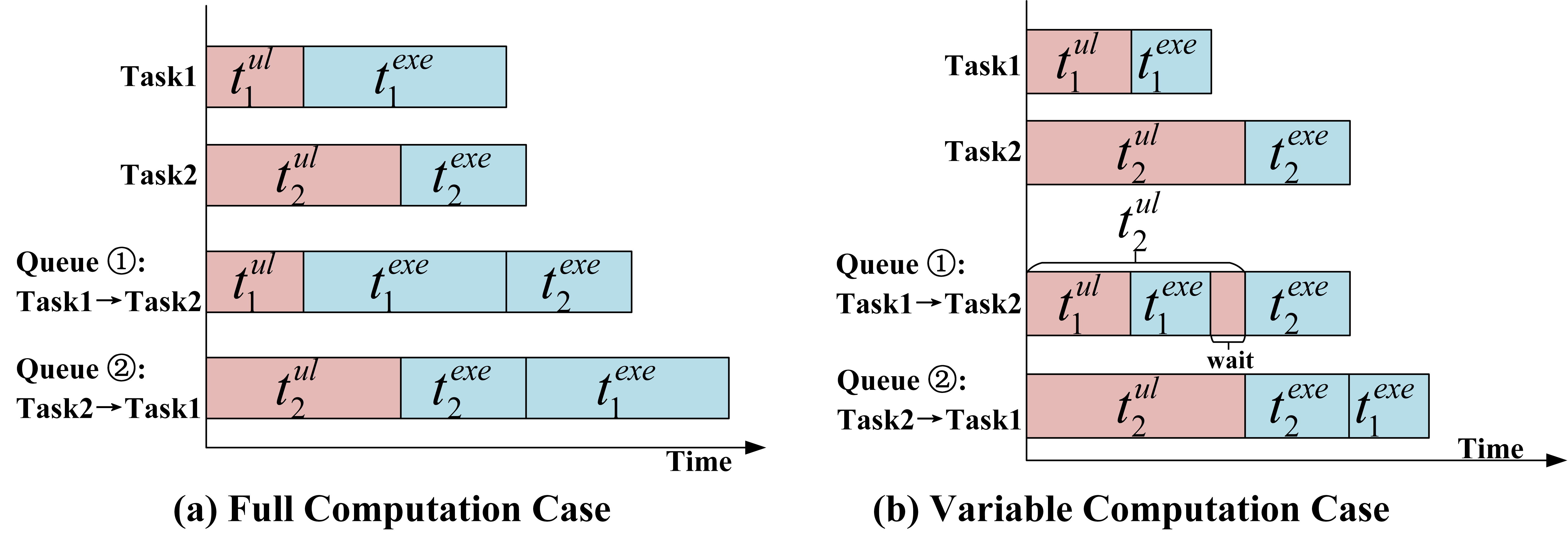}}
\caption{Basic tasks scenario in execution. (a) Full Computation Case. (b) Variable Computation Case.}
\label{fig2}
\end{figure}

Here, the cumulative sum is processed by matrix multiplication. To analyze the computing queue of MEC cores, two basic scenarios named full computation case and variable computation case of the core are analyzed. When Task 1 and Task 2 are in the computing queue of a core, the basic scenario is decided by the relationship of their arrival time $t_1^{ul}$,$t_2^{ul}$, and computing delay $t_1^{exe}$,$t_2^{exe}$:

\begin{itemize}
    \item 	When the arrival time of Task 2 is less than the full delay $t_1^{ul}+t_1^{exe}$ of Task 1, the base scenario is a Full Computation Case. In this scenario, if the computing queue is Queue \textcircled{\small{1}}, Task 2 will have reached the core when Task 1 is computed and the total delay for the core is $t_1^{ul}+t_1^{exe}+t_2^{exe}$. Similarly, if the computing queue is Queue \textcircled{\small{2}}, Task 1 will have reached the core when Task 2 is computed, and the total delay of the core is $t_2^{ul}+t_2^{exe}+t_1^{exe}$. The schematic diagram of the core calculation is shown in Fig.~\ref{fig2} (a).

    \item 	When the arrival time of Task 2 is large than the full delay of Task 1, the base scenario is a Variable Computation Case. In this scenario, if the computing queue is Queue \textcircled{\small{1}}, Task 2 will still be transferred when Task 1 is computed and the total delay for the core is $t_2^{ul}+t_2^{exe}$. Conversely, if the computing queue is Queue \textcircled{\small{2}}, Task 1 will have reached the core when Task 2 is computed, and the total delay of the core is $t_2^{ul}+t_2^{exe}+t_1^{exe}$. The schematic diagram of the core calculation is shown in Fig.~\ref{fig2} (b).    
\end{itemize}

Building on the insights gained from the basic scenarios, we generalize and extend our findings to a broader range of computational situations. The recursive function $F_{rec}$ is defined to realize the delay recursion of tasks calculation offloading in core k whose operation can be summarized as comparative analysis and iterative computation. The arrival time of tasks offloaded from terminals to core k is denoted as a vector $\boldsymbol{t_k^{ul}}$, the computing delay of cores is denoted as a vector $\boldsymbol{t_k^{exe}}$, and the core calculation sequence is $\boldsymbol{q^k}$. Therefore, the actual calculation time $T_k^{core}$ of the core $k$ satisfies:

{\scriptsize
\begin{align}
T_k^{core}=F_{rec}\left((F_{rank}\left(\boldsymbol{t_k^{ul}},\boldsymbol{q^k}\right),F_{rank}\left(\boldsymbol{t_k^{exe}},\boldsymbol{q^k}\right)\right)
  \label{eq4}
\end{align}
}
Similarly, the terminal delay $T_i^{loc}$ can be denoted by terminal task queue $\boldsymbol{q_i^{loc}}$ and vector of terminal’s computing delay $\boldsymbol{\widetilde{t_i^{exe}}}$, that is:

{\scriptsize
\begin{align}
T_i^{loc}=F_{rec}\left((F_{rank}\left(\boldsymbol{0},\boldsymbol{q^{loc}_i}\right),F_{rank}\left(\boldsymbol{\widetilde{t_i^{exe}}},\boldsymbol{q^{loc}_i}\right)\right)
  \label{eq5}
\end{align}
}

The total system delay $T_{sys}$ was inferred by the max value of $T_i^{loc}$ and $T_k^{core}$ which is the delay in each terminal and offloading delay of each MEC computing core respectively. In other words, it’s a Min-Max problem. Therefore, the total system delay $T_{sys}$ is:

{\scriptsize
\begin{align}
    T_{sys}=\max\left(T_1^{core},T_2^{core},...,T_1^{loc},T_2^{loc},...\right)
  \label{eq6}
\end{align}
}

\subsubsection{Power Consumption}\label{Power Consumption}

In CCMSM, the total energy consumption of the system is the sum of the computing part and the offloading part, so we convert this problem into the minimize energy consumption. The computing energy consumption of each task in terminal $i$ can be formulated as:

{\scriptsize
\begin{align}
    e_{ij}^{exe} =
        \begin{cases}
            k_i(f_i^{loc})^2c_{ij}, & \forall w_{ij}\in Q^0 \\
            k_i(f_k^{core})^2c_{ij}, & \forall w_{ij}\in\bigcup_{k=1}^{K}Q^k
        \end{cases}
  \label{eq7}
\end{align}
}

Where $k_i$ is the power constant related to the CPU in terminals and cores. Similarly, the transfer energy consumption of the tasks in terminal $i$ can be defined as:

{\scriptsize
\begin{align}
    e_{ij}^{exe} =
        \begin{cases}
            \frac{p_i^{loc}u_{ij}}{ W^{ul}\log_2\left(1+\frac{p_i^{loc}}{N_0W^{ul}}\right)}, & \forall w_{ij} \in \bigcup_{k=1}^{K}Q^k \\
            0, & \forall w_{ij}\in Q^0
        \end{cases}
  \label{eq8}
\end{align}
}

In Eq. \eqref{eq8}, the transfer energy consumption of the task computed in terminals is recorded as $0$. Due to the independence of the computing energy consumption and computing queue, the task computing energy consumption of the system can be expressed as $E^{exe}=\sum\sum e_{ij}^{exe}$. As the transfer energy consumption is also independent of the offloading queue, the task transfer energy consumption of the system can be denoted as $E^{ul}=\sum\sum e_{ij}^{ul}$. Therefore, the total energy consumption of the system can be expressed as $ E=E^{exe}+E^{ul}$.

\subsubsection{Utility Index}\label{Utility Index}
To simultaneously optimize the system's energy consumption and delay, a utility index is designed to couple them. By introducing coefficient $\lambda$ to adjust the optimization’s preference, the multi-objective problem is transformed into a single-objective problem, which is denoted as:

{\scriptsize
\begin{align}
    U=\lambda E+(1-\lambda)T_{sys}
  \label{eq12}
\end{align}
}
Where $\lambda\in[0,1]$. Finally, the optimization in CCMSM is formulated as:

{\scriptsize
\begin{align}  
	\underset{\boldsymbol{q^k},\boldsymbol{q_{i}},\boldsymbol{q^{loc}}, p^{loc}_i}{\mathrm{minimize}}  \quad& U \label{eq13}\\
	\notag
	\text{subject to:}  \quad& p^{loc}_{min} \leq p^{loc}_i \leq p^{loc}_{max}  \tag{10a}\label{eq13:con1}\\
	\quad& \bigcap_{k=0}^{K}Q^k=\emptyset \tag{10b}\label{eq13:con2}\\
	\quad& \bigcup_{k=0}^{K}Q^k=\bigcup W_i \tag{10c}\label{eq13:con3}\\
	\quad& Q_i^{loc}=Q^0\bigcap W_i \tag{10d}\label{eq13:con4}
\end{align}
}

Where the Eq. \eqref{eq13:con1} represents the transmission power constraint of the terminal, the maximum and minimum transmit power are denoted as $p^{loc}_{max}$ and $p^{loc}_{min}$. Eq.\eqref{eq13:con2} to \eqref{eq13:con4} represents the constraints of the offloading strategy, which satisfy the independence of each task and the uniqueness of each strategy.

\subsubsection{Decoupling of Optimization Objectives}

Since the optimization variables $\boldsymbol{q_i^{ul}}\boldsymbol{q_i^{loc}}$ and $\boldsymbol{q^k}$ are discrete, while $p^{AP}$ is continuous, the optimization in section \ref{Utility Index} is a Mixed Integer Nonlinear Programming problem (MINLP)\cite{b27}, which is usually considered as an NP-hard problem \cite{b28}. Therefore, it is necessary to reform and simplify the original problem. Generally, the discrete variable is relaxed into a continuous variable before optimization with a single algorithm. The solution of the continuous variable can be obtained by bringing the solution of the discrete variable into the original problem\cite{b31}. However, relaxation will convert the original problem's infeasible regions into continuous regions, leading to convergence to a local optimum instead of a global one, which is unacceptable for task scheduling.

Fortunately, when the optimal solution of MINLP problem has apparent characteristics\cite{b32}, applying a greedy strategy can lead to a good solution. To simplify the optimization, we apply a greedy strategy to decouple it into two parts: the optimization of offloading strategy and power allocation. The first part is formulated as P1:

{\scriptsize
\begin{align}
    \text{P1:}\quad 
	\underset{\boldsymbol{q^k},\boldsymbol{q_{i}},\boldsymbol{q^{loc}}, }{\mathrm{minimize}}  \quad& U \label{eq14}\\
	\notag
	\text{subject to:}
	\quad& \bigcap_{k=0}^{K}Q^k=\emptyset \tag{11a}\label{eq:con2}\\
	\quad& \bigcup_{k=0}^{K}Q^k=\bigcup W_i \tag{11b}\label{eq:con3}\\
	\quad& Q_i^{loc}=Q^0\bigcap W_i \tag{11c}\label{eq:con4}
\end{align}
}

Then, the optimized offloading strategy will be used as a parameter to optimize the power, which is denoted as P2:

{\scriptsize
\begin{align}
    \text{P2:}\quad 
    \underset{p^{loc}_i}{\mathrm{minimize}} \quad& U \label{eq15}\\
	\notag
	\text{subject to:}  \quad& p^{loc}_{min} \leq p^{loc}_i \leq p^{loc}_{max}  \tag{12a}\label{eq:con1}
\end{align}
}
\section{Algorithm}
\subsection{Encoding}
In the CCMSM, the GA is applied for P1 and SSA for P2. It should be noted that the encoding process, which reflects the core principles of the scheduling and enables the algorithm to be effectively integrated with the model, is the focus of the proposed GA-SSA. Thus, a genetic encoding scheme is proposed to adapt the strategy to the optimization process of the algorithm.

As the utility index will be affected by the offloading queue and the mapping of computing cores, which is similar to the gene in structure in P1. They are taken as the variables of the algorithm and participate in the iterative optimization process.

\begin{figure}[htbp]
\centerline{\includegraphics[height=4cm,width=0.48\textwidth]{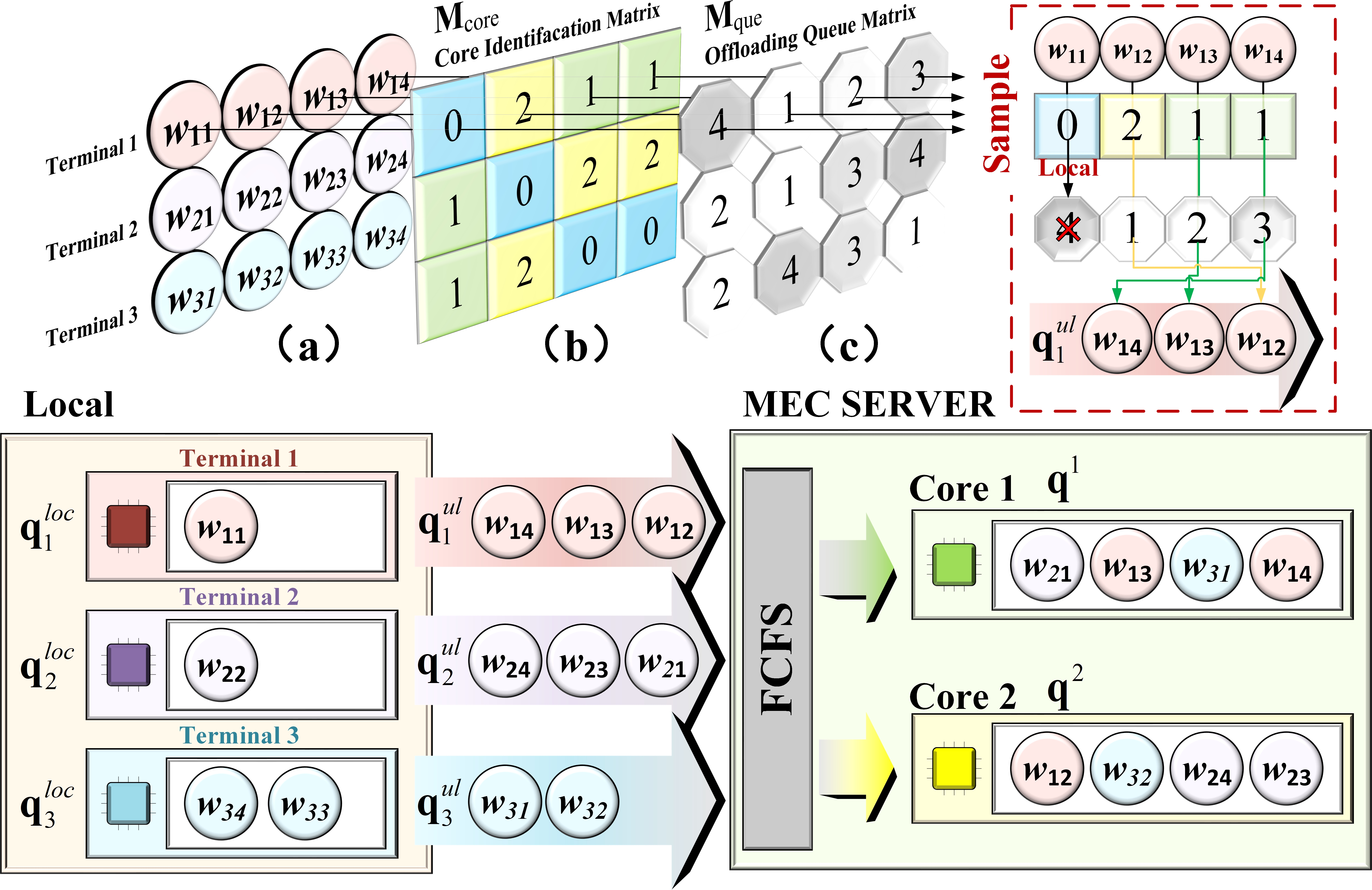}}
\caption{Example diagram of encoding in offloading}
\label{fig3}
\end{figure}

Fig.~\ref{fig3} shows a typical process of offloading encoding. First, the task $w_{ij}$ of terminals will be arranged in rows and taken to form a matrix $\boldsymbol{M_{que}}$ in collaboration with the offloading queue. Simultaneously, tasks are initialized on different cores to form a core identity matrix $\boldsymbol{M_{core}}$. E.g., considering that task $w_{11}$ is computed in the terminal, determined by the $M_{core,1,1}=0$, the $\boldsymbol{q_1^{ul}}$ of Terminal 1 is $w_{12}\rightarrow w_{13}\rightarrow w_{14}$ which is decided by the $M_{que,1,:}$. As for others, the scheduling is the same as Terminal 1. By repeating the above operation, the task is $w_{21}\rightarrow w_{13}\rightarrow w_{31}\rightarrow w_{14}$ in core 1 and $w_{12}\rightarrow w_{32}\rightarrow w_{24}\rightarrow w_{22}$ in core 2.

When the offloading queue and the mapping relation of cores are determined, the computing queue of the terminal and the core needs to be further determined. As the Eq. \eqref{eq5} and \eqref{eq8} in section \ref{Optimization Goals}, the computing queue does not affect the terminal delay and energy consumption of the system, which will be ignored in modeling. In contrast, the computing queue of cores is critical to the core computing delay which should be minimized. Typically, the First-Come-First-Serve (FCFS) strategy is commonly considered to be optimal, as it minimizes the computing delay and ensures that no task is indefinitely delayed which simplifies and accelerates the algorithm. The proof is presented in the next section.

\subsection{Computing Queue Generation}

To get the computing delay of the core, we need the following steps to determine that the computing queue is governed by the FCFS strategy of cores.

\begin{lemma}
    When it exists Task 1, Task 2 in the core, with the arrival time satisfying $t_{ul1}\le t_{ul2}$, the total computing delay of the core is the minimum if the computing strategy is FCFS.
\end{lemma}

\begin{proof}
	When the computing queue is Task 1$\rightarrow$Task 2, the total computing delay is:
 
	{\scriptsize
        \begin{align}
		T_{12}=\begin{cases}
			t_{ul1}+t_{exe1}+t_{exe2}, &t_{ul2}<t_{ul1}+t_{exe1}  \\
			t_{ul2}+t_{exe2}, &t_{ul2}\geq t_{ul1}+t_{exe1}
			\end{cases}
	\end{align}
	}
 
	Similarly, when the computing queue is Task 2$\rightarrow$Task 1, the computing delay is:
 
	{\scriptsize
	\begin{align}
		T_{21}=\begin{cases}
			t_{ul2}+t_{exe2}+t_{exe1}, &t_{ul1}<t_{ul2}+t_{exe2}  \\
			t_{ul1}+t_{exe1}, &t_{ul1}\geq t_{ul2}+t_{exe2}
		\end{cases}
	\end{align}
	}
 
	Since $t_{ul1}<t_{ul2}, t_{ul1}\geq t_{ul2}+t_{exe2}$ does not hold, $T_{21}=t_{ul2}+t_{exe2}+t_{exe1}$. For the reason that $t_{ul2}+t_{exe2}+t_{exe1}\geq\max(t_{ul1}+t_{exe1}+t_{exe2},t_{ul2}+t_{exe2})$ always hold, that is $T_{21}\geq T_{12}$. Therefore, \textit{Lemma 1} is held.
	
\end{proof}

\begin{lemma}
    When there are three tasks, Task 1, Task 2, and Task 3, assigned to a core, only Task 3 is determined to be computed at the end, while either Task 1 or Task 2 needs to be selected for computation. The total delay $T_{13}$ of Task 1$\rightarrow$Task 3 always less than or equal to the total delay $T_{23}$  of Task 2$\rightarrow$Task 3 if $t_{ul1}+t_{exe1}\le t_{ul2}+t_{exe2}$.
\end{lemma}

\begin{proof}
	When the computing queue is Task 1$\rightarrow$Task 2, the computing delay is:
 
	{\scriptsize
        \begin{align}
		T_{13}=\begin{cases}
			t_{ul1}+t_{exe1}+t_{exe3}, &t_{ul3}\le t_{ul1}+t_{exe1}  \\
			t_{ul3}+t_{exe3}, &t_{ul3}\geq t_{ul1}+t_{exe1}
		\end{cases}
	\end{align}
        }
	
	Similarly, when the computing queue is Task 1$\rightarrow$Task 3, the computing delay is:
 
        {\scriptsize
	\begin{align}
		T_{23}=\begin{cases}
			t_{ul2}+t_{exe2}+t_{exe3}, &t_{ul3}\le t_{ul2}+t_{exe2}  \\
			t_{ul3}+t_{exe3}, &t_{ul3}\geq t_{ul2}+t_{exe2}
		\end{cases}
	\end{align}
	}
 
	Traverse the possible situations of $T_{13}$ and $T_{23}$, when $t_{ul3}<t_{ul1}+t_{exe1}$ and $t_{ul3}<t_{ul2}+t_{exe2}$ are hold, then we have $t_{ul1}+t_{exe1}+t_{exe3}<t_{ul2}+t_{exe2}+t_{exe3}$, it’s obviously that $T_{13}<T_{23}$. 
 When $t_{ul3}<t_{ul1}+t_{exe1}$ and $t_{ul3}\geq t_{ul2}+t_{exe2}$ hold, $t_{ul1}+t_{exe1}<t_{ul2}+t_{exe2}$ is contrary to the definition, which does not hold. 
 When $t_{ul3}\geq t_{ul1}+t_{exe1}$ and $t_{ul3}<t_{ul2}+t_{exe2}$, we can easily get $t_{ul1}+t_{exe1}+t_{exe3}<t_{ul2}+t_{exe2}+t_{exe3}$, that is $T_{13}<T_{23}$.
 When $t_{ul3}\geq t_{ul1}+t_{exe1}$ and $t_{ul3}\geq t_{ul2}+t_{exe2}$ hold, we can get $t_{ul3}+t_{exe3}=t_{ul3}+t_{exe3}$, that is $T_{13}=T_{23}$. Therefore, \textit{Lemma 2} is held, and it can be considered that the core computing delay depends on the preceding tasks when the last tasks in the computing queue are the same.
	
\end{proof}

\begin{Corollary}
   When there are three tasks, Task 1, Task 2, and Task 3, assigned to a core and all need to be computed while the arrival time of them satisfying $t_{ul1}<t_{ul2}<t_{ul3}$, the total computing delay of the core is the minimum if the computing strategy is FCFS.
\end{Corollary}

\begin{proof}
The queue of tasks in the core can be: \\
 {
        \scriptsize
	\begin{tabular}{ccc|cc}
		&Case1:&Task 1$\rightarrow$Task 2$\rightarrow$Task 3 &Case4:&Task 3$\rightarrow$Task 1$\rightarrow$Task 2  \\
		&Case2:&Task 2$\rightarrow$Task 1$\rightarrow$Task 3 &Case5:&Task 2$\rightarrow$Task 3$\rightarrow$Task 1  \\
		&Case3:&Task 1$\rightarrow$Task 3$\rightarrow$Task 2 &Case6:&Task 3$\rightarrow$Task 2$\rightarrow$Task 1  \\
	\end{tabular}
 }

According to Lemma 2, with the same ending task, the computing delay depends on the preceding task. So, case 2, case 4, and case 6 are inferior solutions. By analyzing the remaining cases, the total possible computing delay is

\begin{center}
 
{\scriptsize
	\begin{tabular}{c|c}
		\hline
		Case 1&Case 3\\
  \hline
$t_{ul1}+t_{exe1}+t_{exe2}+t_{exe3}$ & $t_{ul1}+t_{exe1}+t_{exe3}+t_{exe2}$ \\
$t_{ul2}+t_{exe2}+t_{exe3}$ & $t_{ul3}+t_{exe3}+t_{exe2}$\\
$t_{ul3}+t_{exe3}$  &$\bcancel{t_{ul2}+t_{exe2}}$\\
\hline

Case 5&\\
\hline
$t_{ul2}+t_{exe2}+t_{exe3}+t_{exe1}$&\\
$t_{ul3}+t_{exe3}+t_{exe1}$&\\
$\bcancel{t_{ul1}+t_{exe1}}$&\\
\hline
	\end{tabular}
 }
\end{center}

Assuming $t_{ul1}+t_{exe1}+t_{exe3}+t_{exe2}\geq t_{ul3}+t_{exe3}+t_{exe2}$ holds in case 3, the total time delay calculated is $t_{ul1}+t_{exe1}+t_{exe3}+t_{exe2}$, which is obviously greater than the three possible computing delay in case1, thus, case 1 is superior to case 3; 
Assuming $t_{ul1}+t_{exe1}+t_{exe3}+t_{exe2}<t_{ul3}+t_{exe3}+t_{exe2}$is held in case 3, the computing delay of case 3 is $t_{ul3}+t_{exe3}+t_{exe2}$, which is obviously greater than the three possible computing delay of case1, thus case1 is also superior to case3. 
The comparison between case 5 and case 1 is similar. Generally, case 1 is superior to case 5 and \textit{Corollary 1} is proved.

\end{proof}

\begin{Corollary}
 It can be deduced that the total computational delay of FCFS is minimized with $N+1$ tasks in the core if the $N$ tasks-case hold.
\end{Corollary}

\begin{proof}
 When the computing queue of all tasks is classified as ending with Task $1 \sim N+1$, we can obtain $N+1$ types of sequences, and each type contains N sequences of permutations. Taking Task$i$ as an example $(2\leq i\leq N)$,  the optimal sequence in this type can be obtained from \textit{Lemma 2}; when it exists that FCFS minimize the total delay of core computation with $N$ tasks, there will be a Task 1$\rightarrow$...$\rightarrow$Task$i-1$$\rightarrow$Task$i+1$$\rightarrow$...$\rightarrow$Task$N+1$$\rightarrow$Task$i$ s the optimal queue ending with Task$i$, so only $N+1$ possible cases need to be compared.
\begin{center}

{\scriptsize
	\begin{tabular}{ccc}
		&Case 1:& Task 2$\rightarrow$Task 3$\rightarrow$…$\rightarrow$TaskN+1$\rightarrow$Task 1  \\
		&Case 2:& Task 1$\rightarrow$Task 3$\rightarrow$…$\rightarrow$TaskN+1$\rightarrow$Task 2  \\
		&$\cdots$& $\cdots$  \\
		&FCFS:& Task 1$\rightarrow$Task 2$\rightarrow$…$\rightarrow$TaskN$\rightarrow$TaskN+1  \\
	\end{tabular}
}
\end{center}
When comparing the case $2 \sim N$ with FCFS, Task 1\\$\rightarrow$...$\rightarrow$Task$i-1$ in case $i$$(2\leq i\leq N)$ is considered as one overall task, and Task$i+1$$\rightarrow$...$\rightarrow$Task$N+1$ is another overall task, the problem degenerates to a 3-task model, and that FCFS is the sequence with minimum total core computation delay can be concluded from the \textit{Corollary 1}. In the comparison between case 1 and FCFS, Task 2$\rightarrow$Task 3$\rightarrow$...$\rightarrow$Task$N+1$ is considered as an overall task, then the problem degenerates to a 2-task model, and by \textit{Lemma 1}, FCFS leads to minimum computing delay.
\end{proof}

Therefore, the FCFS sequence bringing the minimum computation delay is been proved. And the computation strategy of the core should be FCFS. 

\section{Evaluation}

\begin{figure*}[tbp]
\centerline{\includegraphics[width=0.9\textwidth]{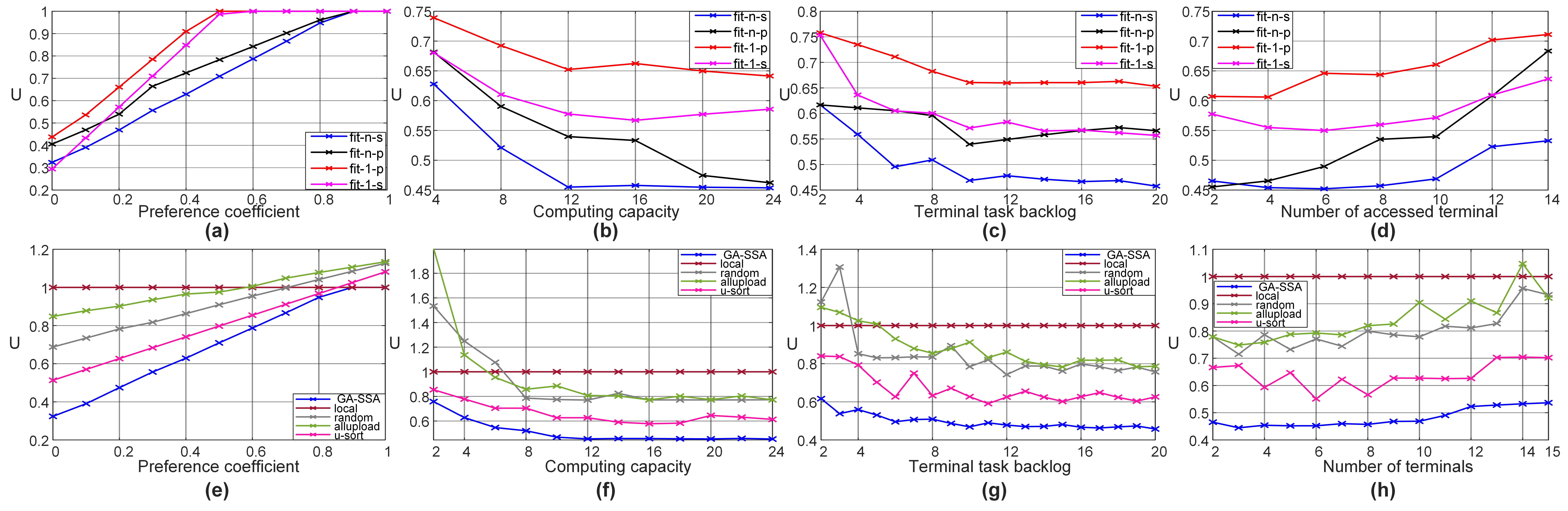}}
\caption{Simulation of 4 different models and algorithms. (a) relation between utility index U and preference coefficient $\boldsymbol{\lambda}$. (b) relation between utility index U and the number of cores (equivalent computing capacity). (c) relation between utility index U and the number of tasks (Terminal task backlog). (d) relation between utility index U and the number of accessed terminals. (e) relation between utility index U and preference coefficient $\boldsymbol{\lambda}$. (f) relation between utility index U and number of cores (equivalent computing capacity). (g) relation between utility index U and number of tasks (Terminal task backlog). (h) relation between utility index U and number of accessed terminals.}
\label{fig4}
\end{figure*}

\subsection{Model Comparison}

The experiment concentrates on the optimization level of different MEC models using GA-SSA algorithm. The following baselines are included in the comparison: (1) Multi-core parallel scheduling model under limited channel resources (n-p), (2) Single-core parallel scheduling model with limited channel resources (1-p), (3) Single-core sequential scheduling model with limited channel resources (1-s), (4) Proposed CCMSM (n-s).

Among them, the computing capacity of the MEC server in the 4 models is equal. The experiment is carried out in 4 different scenarios which focus on preference, computing capacity, terminal tasks backlog and access count.

In general, CCMSM showed superiority compared with other models and exhibits better robustness in every MEC deployment scenario in the experiment. For most optimization preferences, the proposed CCMSM can achieve a wider preference range with $\lambda \in [0,0.9]$ in Fig.~\ref{fig4} (a). Therefore, to fall in the valid range of models, $\lambda$ is uniformly set to 0.4 in other scenarios. Fig.~\ref{fig4} (b) demonstrates the superiority of the CCMSM becomes increasingly evident with less computing capacity with the model outperforming the suboptimal model by $7.71\%$ in terms of optimization level. A similar trend is observed in Fig.~\ref{fig4} (c-d), where the advantage of the CCMSM over the suboptimal model increases to $17.81\%$ and $16.29\%$ as the backlog and access increase, respectively. These phenomena indicate that the proposed CCMSM exhibits resilient scalability for traffic spikes with computing resource scarcity with an overall improvement in optimization level of $13.94\%$ over the suboptimal model.

\subsection{Algorithm in CCMSM}

Although CCMSM is more preeminent than other models and exhibits greater verisimilitude, it has higher complexity in optimization. We propose GA-SSA as a competent algorithm and compare it with the following offloading schemes:(1) First come first serve (FCFS), (2) Offloading all (allupload), (3) Random Offloading (random), (4) Terminal Computing (local), (5) Proposed GA-SSA (GA-SSA).

The simulation results depicted in Fig.~\ref{fig4} (e-h) demonstrate that the GA-SSA exhibits superior performance over other offloading schemes. To facilitate the follow-up experiment, the preference coefficient is set to a median value of 0.5. Notably, Fig.~\ref{fig4} (f-h) shows a slow fluctuation trend of the GA-SSA demonstrating superior stability in optimization which is respectively reflected in the standard deviation of $0.087$, $0.048$ and $0.027$. Therefore, the GA-SSA algorithm represents a competent and robust algorithm for optimization in CCMSM.

\section{Conclusion}

In this paper, we proposed a channel-constrained multi-core sequential model (CCMSM) for congestion management (CM) of 6G vehicular networks. Furthermore, a novel utility index was designed as the objective for cost optimization, and the decoupled problem was solved using a combination of genetic algorithm and sparrow search algorithm (GA-SSA), which proved to be a competent and effective approach. The simulation demonstrated that CCMSM can achieve a more refined optimization level in several MEC deployment scenarios compared with the existing models. Additionally, it exhibited better robustness and resilient scalability for traffic spikes.

\end{document}